\documentclass{article}

\usepackage{graphicx} 

\usepackage{amsmath}
\usepackage{amssymb}
\usepackage{amsthm}

\newtheorem{proposition}{Proposition}

\usepackage{hyperref}
\usepackage{cleveref}

\usepackage[giveninits=true,maxbibnames=99]{biblatex}
\usepackage{csquotes}
\let\cite\parencite

\title{Local search for valued constraint satisfaction parameterized by treedepth}
\author{Artem Kaznatcheev}

\begin{document}

\maketitle

\begin{abstract}
Sometimes local search algorithms cannot efficiently find even local peaks.
To understand why, I look at the structure of ascents in fitness landscapes from valued constraint satisfaction problems (VCSPs) parameterized by the treedepth of their constraint graphs.
There are existing constructions of VCSPs with logarithm treedepth that represent fitness landscapes where all ascents are exponential from some initial assignment.
I improve these bounds by showing that with loglog treedepth, superpolynomial ascents exist; and for polylog treedepth, there are initial assignments from which all ascents are superpolynomial.
My hope is that these examples of sparse VCSPs can help us better understand the barriers to efficient local search.
\end{abstract}

\section{Introduction}

\textcite{repCP} showed that no long ascents exist in fitness landscapes corresponding to binary Boolean valued constraint satisfaction problems (VCSPs) with constraint graphs that are trees.
However, they also showed that moving to three-valued VCSPs or increasing the constraint graph to pathwidth 2, allows the existence of exponentially long ascents in the corresponding fitness landscapes.
Their construction of this pathwidth 2 example was a slight simplification of a max degree 4 construction of MAX-CUT by \textcite{MaxCutDeg4}.
\textcite{tw7} showed that not only do such long ascents exist, but that for pathwidth 7 constraint graphs, even reasonable local search dynamics like greedy follow a steepest ascent of exponential length to a local fitness peak.
\textcite{pw4} subsequently reduced the parameter for an exponential steepest ascents to just pathwidth 4.
Both sets of authors were unaware of \textcite{HL_Steepest}'s much earlier construction that had pathwidth 3 (as noted by \textcite{effective-and-efficient}).
Recently, \textcite{pw2MSc} and \textcite{slow-greed} reduced the `simplicity' of construction of VCSPs that produce exponential steepest ascents to the lowest possible value of pathwidth 2.

The above results seem to stand in contrast to our knowledge of solving for the global-optimum of VCSPs using non-local-search methods.
In particular, the global optimum of a VCSPs with a bounded treewidth constraint graph can be found in polynomial time~\cite{BB73,VCSPsurvey}.
In an earlier 2024 version of this paper~\cite{td_short}, I wanted to address this discrepancy.
But I made a number of unfixable errors in one of the central proofs.
In particular, I claimed to show that in a fitness landscape corresponding to a bounded treewidth VCSP, there always exists a polynomial length assent to some local optimum.
I claimed to do this by showing the stronger result that VCSPs with logarithmic treedepth always have a polynomial length ascent to a local optimum.
My claims were false.
In fact, these result could not have been true, and I could have known this from better study of existing constructions.

Specifically, there were existing constructions where all ascents are exponential from some initial assignment that have bounded pathwidth and thus also log-bounded treedepth. 
Although \textcite{MaxCutDeg4} was focused on reducing the maximum degree of Max-Cut instances, their construction for some ascent exponential had pathwidth 2.
And they claimed that their some-ascent-exponential construction can be converted to an all-ascents-exponential construction by growing the gadgets that they used from five variables to twenty-eight.
In 2024, \textcite{MaxCutSimple} provided a simpler Max-Cut construction with maximum degree four and all ascents exponential.
Neither of these two Max-Cut papers were focused on the pathwidth of their constructions.
But given that the constructions were a chain of finite gadgets, it is clear that they had finite pathwidth and this is enough to show my old claim to be impossible.
Finally, in 2026, \textcite{VCSPpw3_allExp} showed that both of the prior Max-Cut constructions had pathwidth $\geq 4$, and provided a ternary Boolean VCSP constructions that reduces the pathwidth to three while having all ascents exponential from some designated initial assignment.
From these constructions, we know that it is not possible for a short ascent to exist from all initial assignments in fitness landscapes represented by VCSPs of constant treewidth or logarithmic treedepth.

For VCSPs of constant treedepth, \textcite{star_VCSPs} still conjecture that a polynomial ascent always exists.
However, they also produce a binary Boolean VCSP of treedepth $3$ that has an exponential ascent.
This means that my stronger conjecture from the 2024 version of this paper~\cite{td_short} on all ascents being short is false: not all ascents are polynomially long for VCSPs of constant treedepth.

So although the main upper-bound theorem of this paper was false, the two main lower-bound constructions are still useful.
In \cref{sec:all_long}, I give an example of how to build a family of VCSP instances of polylog treedepth where all ascents are superpolynomial.
And in \cref{sec:some_long}, I give an example of how build a family of VCSP instances of loglog treedepth where some ascent is superpolynomial.
This construction has been superseded by \textcite{star_VCSPs}'s VCSP of treedepth 3 with an exponential ascent, but it is useful as an example of a different kind of construction.

\section{Background}

A valued constraint satisfaction problem (VCSP) instance on $n$ variables consists of $n$ domains $D_i$ and a finite set of valued constraints $\mathcal{C} = \{C_S\}$. 
If $|D_i| = 2$ for all $i \in [n]$ then we say that the VCSP is Boolean.
Each valued constraint $C_S$ with \textbf{scope} $S \subseteq [n]$ is a function $C_S: \prod_{i \in S} D_{i} \rightarrow \mathbb{Z}$.
Each \textbf{assignment} $x \in D_1 \times \cdots \times D_n$ has an associated \textbf{fitness} (or \textbf{value}) $f(x)$ given by $f(x) = \sum_{C_S \in \mathcal{C}} C_S(x[S])$ where $x[S]$ is the substring of assignment $x$ with domain indexes in $S$.
The \textbf{fitness landscape} from a VCSP instance is the above fitness function together with a notion of when two assignments are adjacent.

Two assignments $x,y$ are \textbf{adjacent} if there exists a unique $k \in [n]$ such that $x_k \neq y_k$ and $x_i = y_i$ for all other $i \in [n] - \{k\}$.
This unique $k$ is the \textbf{step-index} between $x$ and $y$.
$N(x) \subseteq  D_1 \times \cdots \times D_n$ is the set of all assignments that are adjacent to $x$.
A \textbf{step-sequence} $p = x^0,x^1,\ldots,x^T$ is a sequence of $T + 1$ assignments such that $x^{t + 1} \in N(x^t)$.

A \textbf{(local) solution} to a VCSP $\mathcal{C}$ is an assignment $x^*$ such that for all $y \in N(x^*)$ $f(x^*) \geq f(y)$.
An \textbf{ascent} $p = x^0,x^1,\ldots,x^T$ of length $T$ in a VCSP $\mathcal{C}$ is a step-sequence of $T + 1$ assignments such that $f(x^t) < f(x^{t + 1})$ and $x^T$ is a local solution to $\mathcal{C}$.
An ascent $p$ is \textbf{step-steepest} when for each $t$ if $k$ is the step-index between $x^t$ and $x^{t + 1}$ and $y$ is any string such that $y[[n] - \{k\}] = x[[n] - \{k\}]$ then $f(x^{t + 1}) \geq f(y)$.
In other words, if a step in $p$ changes a variable $x_k$ then it makes the most fitness-increasing change possible to that variable.
Given a partial ordering $\prec$ of $[n]$, a \textbf{$\prec$-ordered ascent} is an ascent $p = x^0,x^1,\ldots,x^T$ with $k^t$ as the step-index between $x^t$ and $x^{t+1}$ such that for all $j \prec k$ and $y \in N(x^t)$, if $j$ is the step-index between $x^t$ to $y$ then $f(y) \leq f(x^t)$.
In other words, each step in $p$ flips a $\prec$-minimal index of those that are able to flip. 

Given a VCSP $\mathcal{C}$, the corresponding \textbf{constraint graph} $G_\mathcal{C} = ([n],E)$ has $\{i,j\} \in E$ if and only if there exists some constraint $C_S \in \mathcal{C}$ such that $\{i,j\} \subseteq S$.
Given a tree $T=([n],P)$ rooted at $r_0$, let $\prec_T$ be the descendant relationship for $T$.
In other words, $r \prec_T s$ if the path from $r$ to $r_0$ in $T$ passes through $s$ and $r \neq s$.
$T$ is a \textbf{treedepth-decomposition} of $G_\mathcal{C}$ if for every $uv \in E$ either $u \prec_T v$ or $v \prec_T u$.
The \textbf{treedepth} $\mathrm{td}(G_\mathcal{C})$ is the minimal $\mathrm{height}(T)$ over all treedepth-decompositions of $G_\mathcal{C}$.
Treedepth is closely related to other sparsity metrics like treewidth~\cite{td_chapter}.
In particular, graphs with bounded treewidth have at most logarithmic treedepth~\cite{BGHK95,td_chapter}.

\section{All ascents are long with polylog treedepth}
\label{sec:all_long}

Polylogarithmic threedepth allows for all ascents to be long.

\begin{proposition}
There exists a Boolean VCSP on $n$ variables of treedepth $d$, such that -- starting from $0^n$ -- every ascent has length at least $\frac{9}{64d}2^{d} \cdot n$.
\label{prop:main_tdd_tight}
\end{proposition}

This follows directly from existing results on snake-in-the-box codes~\cite{snakeLowerBound}.
A \textbf{snake} is a connected path in the hypercube where each node on the path, with the exception of the head and tail, has exactly two neighbours that are also in the snake. 
The head and the tail each have only one neighbour in the snake. 
The snake can visit a candidate assignment in the hypercube if the candidate node is connected to the current node and it is not a neighbour of any previously visited node in the snake, other than the current node.
\textcite{snakeLowerBound} showed the there exists snakes on $d$-bits with length lower bounded by $\frac{9}{64}2^d$.

\begin{proof}[Proof of Proposition~\ref{prop:main_tdd_tight}]
Consider a valued constraint on a set $S$ of Boolean variables built in the following way.
Set the fitness of all assignments corresponding to snake code words as their position on the snake path.
Set the fitness of assignments to non-code words as zero. 
This single constraint creates a fitness landscape where every ascent from the assignment corresponding to the snake's head is long.
Set $0^S$ as the snake's head, and call the above a \emph{snake constraint} on $S$.

Now divide the $n$ bits into $n/d$ blocks of $d$ bits each.
On each block, put a snake constraint as above.
Based on \textcite{snakeLowerBound}'s snake construction, each block will require at least $\frac{9}{64}2^{d}$ steps to reach its head.
This will be repeated $n/d$ times, giving us the bound in Proposition~\ref{prop:main_tdd_tight}.
\end{proof}

If we let $S$ contain $d = \mathrm{polylog}(n)$ bits then the above VCSP on $n$ bits will implement a fitness landscape will all ascents of non-polynomial length.
Thus, the VCSP will have no polynomial ascents from some initial assignment while having treedepth that is $\mathrm{polylog}(n)$.

\section{Long ascents with loglog treedepth}
\label{sec:some_long}

 It is important to note that the existence of short ascents does not mean that local search algorithms will always find and follow short ascents in sparse VCSPs.
 For instance, \textcite{repCP}'s Example 7.2 is a binary Boolean VCSP on $4n + 1$ variables with a constraint graph of pathwidth-$2$ that has ascents of length greather than $2^n$.
 This example has logarithmic treedepth, and all other prior examples of long ascents that I know of also have at least logarithmic treedepth.
 To improve this state-of-the-art, I will to construct an binary Boolean VCSP on $n$ variables with treedepth $O(\log \log n)$ that has long ascents.

At treedepth one, the constraint graph has no edges and so can only produce a smooth additive landscape.
This can have at most $n$ flips for $n$ bits.

For treedepth two, we have stars.
Consider a star on $2n + 1$ vertices.
I will divide these vertices into $3$ sets, a left set $L = \{1, \ldots, n\}$ of $n$ vertexes, one center vertex $c = n + 1$, and a right set $R = \{n + 2, \ldots, 2n + 1\}$ of $n$ vertexes.
For every edge $cw$ with $w \in R$, add the constraint:
\begin{equation}
C_{\{c,w\}} = \begin{pmatrix}
0 & 1 \\
1 & 0
\end{pmatrix} 
\quad \text{ where } x_c \text{ selects the row and } x_w \text{ selects the column.}
\end{equation} 
For the center variable, add the unary constraint:
\begin{equation}
C_{\{c\}} = \begin{pmatrix}
0 \\
n + 1
\end{pmatrix}
\quad \text{ where } x_c \text{selects the row.}
\end{equation}
For every edge $uc$ with $u \in L$, if $n - u $ is even then add the constraint:
\begin{align}
C_{\{u,c\}} & = \begin{pmatrix}
0 & 0 \\
2n + 2 & 0
\end{pmatrix}
\quad \text{if } n - u \text{ is even, or} \\
C_{\{u,c\}} & = \begin{pmatrix}
0 & 0 \\
0 & 2n + 2
\end{pmatrix}
\quad \text{if } n - u \text{ is odd.}
\end{align}
Where for both constraints $x_u$ selects the row and $x_c$ selects the column.
Finally, create a unary constraint for each $u \in L$ of $
C_{\{u\}} = \begin{pmatrix}
0 \\
1
\end{pmatrix}$ 
where $x_u$ selects the row.

Now consider a $>$-ordered ascent on this VCSP starting from $x = 0^{2n + 1}$.
This ascent will increase fitness by $1$ at each step for a total length of $n^2 + 4n + 1$ flips.

Finally, for higher treedepth $d$, we will just repeat the pattern of the treedepth one construction.
Specifically, we will have $2^{d}$ blocks of $n$ variables with a $2^{d} - 1$ `center' variables, with one center variable after each block of $n$ variables.
Name these center variables from \emph{right} to \emph{left} as $c_1$, ..., $c_m$ for $m = 2^{d} - 1$.
We will connect each center variable to it's preceding and following block of $n$ variables in the same way as the treedepth $1$ example above, except we will define the weights on $c_k$ as follows.
For $c_1$, we will set the weights as in the star construction above. 
Specifically, the non-zero matrix entries of the binary constraints from left to center, unary constraint on center, and binary constraint from center to right will be $2n + 2, n+1, 1$, respectively.
Finally, for the first block, we will define $w_1 = 2n + 3$ (also $w_0 = 1$).
For $c_{k + 1}$, we will build the weights recursively from $w_k$ by setting the non-zero matrix entries as $2nw_k + 2, nw_k + 1, w_k$ and define $w_{k + 1} = 2nw_k + 3$.
As with the star example, the leftmost block of $n$ variables will have the additional unary constraints $\begin{pmatrix} 0 \\ 1 \end{pmatrix}$.

In the resulting VCSP, as with the star case, the $>$-ordered ascent starting from $0^{2^{d}(n + 1) - 1}$ will increase the fitness by $1$ at each step, with a maximum fitness greater than $n^{2^d}$.
Thus, if we pick $d = \log\log n$, we will have a VCSP with $\log\log$ treedepth and an ascent that is quasipolynomially long.

\section{Conclusion}

Overall, the goal of focusing of parameters like treedepth is to find cases where where can seperate short from long ascents.
Cases where short ascents are guaranteed to exist but reasonable algorithms end up on long ascents.
This will transform the general question of tractability into a question of if specific local search algorithms that are used in practice -- or forced on us by nature when we are modeling processes like biological evolution~\cite{KazThesis,evoPLS,repCP} -- end up on one of the short ascents or one of the long ones.
Overall, these results suggest that the study of sparse VCSPs can help us better understand the barriers to efficient local search and teach us something about the structure of locally optimal combinatorial objects.

\section*{Acknowledgements}

I would like to thank Sukanya Pandey for introducing me to treedepth; 
and Dave Cohen, Peter Jeavons, and Melle van Marle for helpful discussions; and Daniel Dadush for pointing out the fatal errors in my 2026 version of this paper.

\printbibliography

\end{document}